\documentclass[11pt]{article}

\usepackage[T1]{fontenc} 	
\usepackage{lmodern}
\usepackage{color}
\usepackage{amsmath}   
\usepackage{amssymb}   
\usepackage{amsthm}    
\usepackage{amsfonts}

\usepackage{fullpage}
\usepackage{graphicx}

\usepackage[capitalise]{cleveref}
\usepackage{bbm}

\newcommand{\blk}{\color{black}}

\theoremstyle{definition}
\newtheorem{assumption}{Assumption}
\theoremstyle{plain}
\newtheorem{theorem}{Theorem}[section]
\newtheorem{proposition}[theorem]{Proposition}
\newtheorem{corollary}[theorem]{Corollary}

\theoremstyle{remark}
\newtheorem{remark}[theorem]{Remark}

\newcommand{\indicator}{\mathbbm{1}}
\newcommand{\R}{\mathbb{R}}
\newcommand{\W}{\mathbb{W}}

\DeclareMathOperator*{\argmax}{arg\,max}

\usepackage{amsmath,amssymb,amsthm}
\numberwithin{equation}{section}
\usepackage{parskip}
\begin{document}
\title{A Maximum Principle approach to deterministic Mean Field Games of Control with Absorption
\thanks{This work was supported through the profile partnership program between  Humboldt-Universit\"at zu Berlin and Princeton University.}}
\author{Paulwin Graewe\footnote{Deloitte Consulting GmbH, Kurf\"urstendamm 23, 10719 Berlin, Germany; email: pgraewe@deloitte.de} 
	\and 
	Ulrich Horst\footnote{Department of Mathematics, and School of Business and Economics, Humboldt-Universit\"at zu Berlin,
         Unter den Linden 6, 10099 Berlin, Germany; email: horst@math.hu-berlin.de} 
         \and 
         Ronnie Sircar\footnote{Department of Operations Research and Financial Engineering, Princeton University, Sherrerd Hall, Princeton, NJ 08544; email: sircar@princeton.edu}}

\maketitle

\begin{abstract}
We study a class of deterministic mean field games on finite and infinite 
time horizons arising in models of optimal exploitation of exhaustible resources. The main characteristic of our game is an absorption constraint on the players' state process. As a result of the state constraint the optimal time of absorption becomes part of the equilibrium. This requires a 
novel approach when applying Pontyagin's maximum principle. We prove the existence and uniqueness of equilibria and solve the infinite horizon models in closed form. As players may drop out of the game over time, equilibrium production rates need not be monotone nor smooth. 
\end{abstract}



\renewcommand{\baselinestretch}{1.05} \normalsize


\section{Introduction}

This paper establishes existence and uniqueness of equilibrium results for a class deterministic mean field games (MFGs) arising in models of optimal exploitation of exhaustible resources. MFGs provide a convenient tool 
for analyzing complex strategic interactions among many players when each 
individual player has only a small impact on the behavior of other players. In a standard mean field game each player solves a control problem in which an individual player's payoff functional and the dynamics of the controlled state process depend on the empirical distribution of the other players' actions or states. The existence of Nash equilibria in MFGs can be established by solving either a coupled system of {two} partial differential equations (PDEs), a backward Hamilton-Jacobi-Bellman equation determining the players' utility, and a forward Kolmogorov equation determining the evolution of the distribution of states or actions, or by solving a system of McKean-Vlasov forward-backward SDEs where the 
forward component describes the state dynamics and the backward component 
describes the dynamics of the adjoint variable. We refer to the monograph 
\cite{BensoussanFrehseBook} for further background.

In the economics literature MFGs are often called anonymous games. First introduced by Rosenthal \cite{R-1973} and Jovanovic and Rosenthal \cite{JR-1988}, anonymous games have received renewed attention among economists 
in the last  two decades; see \cite{B-1999,DP-2015, H-2005} and references therein. Huang, Malham\'e and Caines \cite{HuangCainesMalhame06} and Lasry and Lions \cite{LasryLions07} independently introduced MFGs into the engineering and mathematical literature. Ever since, MFGs have become an important driver of mathematical innovation, especially in the areas of PDEs and backward stochastic equations. Complementing the theoretical work 
on MFGs, there is a by now substantial literature where anonymous and mean field games have been successfully applied to an array economic and engineering problems, ranging from network security and traffic networks \cite{GT, WYTH} and systemic risk management \cite{CFS-2015}, to portfolio liquidation \cite{ET-2020, FGHP-2018, FHX-2020,HJN-2015} and oil and energy production in competitive markets \cite{chan2014bertrand,ChanSircar15,graber,graber2020mean,gueant}. 

Oligopoly models of markets with a small number of competitive players that compete on the amount of output they produce go back to the classical work of Cournot~\cite{cournot} in 1838. These have typically been static (or one-period) models, where the existence and construction of a Nash equilibrium have been extensively studied. Dynamic Cournot models for energy production in a competitive market have been proposed and analyzed by many authors in recent years; we refer to \cite{LS-survey} for a survey on 
energy production models. In these models market participants are often endowed with limited amounts of exhaustible resources such as oil, coal or 
natural gas that they choose to extract for sale. The models may lead to continuous time single player control problems \cite{fbs}, finite player nonzero-sum differential games \cite{HHS}, or continuum mean field games \cite{chan2014bertrand,ChanSircar15}. In the context of nonzero-sum dynamic games between finitely many players the computation of a solution is a 
challenging problem, typically involving coupled systems of nonlinear PDEs, with one value function per player, and existence theory is sparse. MFGs allow one to handle certain types of competition in the continuum limit of an infinity of small players by solving either systems of PDEs or forward-backward SDEs.

One of the key characteristics of games with exhaustible resources is that resources may run out and change the structure of the game. In a stochastic setting, the strict absorption constraint is often ignored, assuming 
that resource levels fluctuate randomly and recover as soon as they reach 
positive levels again. The situation is very similar in models of optimal 
portfolio liquidation. When asset portfolios are liquidated on behalf of third parties, short positions are not allowed legally. Nonetheless, these constraints are usually ignored for reasons of mathematical tractability. The non-negativity constraint may be binding even in deterministic games as shown in \cite{FGHP-2018}. 

Incorporating non-negativity constraints on the state processes leads to MFGs with absorbing boundaries. The literature on MFGs with absorbing state dynamics is sparse. Models of mean field type without strategic interactions have been applied to model interactive default times by various authors. Graber and Bensoussan \cite{graber} consider a modification of the 
model of Chan and Sircar \cite{chan2014bertrand} on a bounded domain. While the restriction to bounded domains simplifies the mathematical analysis it seems undesirable from an economic perspective. We impose only non-negativity constraints but restrict ourselves to deterministic settings. MFGs with absorbing state constraints in stochastic settings have recently 
been considered by Campi and Fischer in \cite{campi2018}. Their model heavily relies on diffusive state dynamics, though.

We consider a deterministic MFG model of optimal exploitation of exhaustible resources in which players with zero resources drop out of the game before the terminal time in equilibrium. Both finite and the infinite horizon models are studied. As we will see there can be very different equilibrium strategies depending on whether a terminal time is imposed or not. The fact players may exit the game before the terminal time requires a novel approach to the maximum principle associated with an individual player's optimization problem. Due to the non-negativity constraint on the state process, the Hamiltonian depends on the initial condition, while the terminal value of the adjoint process is unknown. The latter is very similar to portfolio liquidation models where the terminal condition of the adjoint process is unknown, due to the liquidation constraint; see \cite{FGHP-2018,FH-2020} for details. In order to overcome this problem we consider the Hamiltonian only up a candidate optimal exploitation time. Obtaining such candidate requires a candidate terminal condition for the adjoint 
process. In our deterministic setting the candidate terminal condition and hence the candidate optimal exploitation time can be given in closed form; the stochastic case is much more involved and left for future work. 

The key observation in solving the MFG is that a certain best response function is independent of the mean field equilibrium. More precisely, we prove that the dynamics of the amount of resource extracted by an individual player up to any {\sl given} time by using a strategy that would be optimal if the optimal exhaustion time was equal to that given time is independent of the {\sl equilibrium} mean production rate. Moreover, we show that the dynamics follows an ODE that depends only the initial distribution of the resource levels and the risk free interest rate. The ODE can be 
solved in closed form. The solution is strictly monotone and the optimal exhaustion time is given by the inverse function. In the infinite horizon case where the equilibrium mean production rate converges to zero as time increases to infinity, this allows us to fully describe the unique 
MFG equilibrium in terms of the said ODE. In the finite horizon case the equilibrium mean production rate at the terminal time is unknown. This results in an additional fixed point problem that can easily be solved. Thus we provide among the very few explicit solutions to MFGs outside the linear-quadratic framework.    

To illustrate our main ideas we first consider in Section \ref{sec:mon} the benchmark case of a single monopolist oil producer. We fully characterize optimal exploitation strategies. In particular, we prove that full exploitation may not be optimal in finite horizon problems. The MFG of optimal extraction is analyzed in Section \ref{sec:game}. We first determine the equilibrium time of exploitation as a function of the competitors' strategies and own initial resources. Subsequently, we determine the equilibrium exploitation times and strategies. We illustrate by two simple examples that equilibrium exploitation rates do need not be monotone, nor smooth.    


\section{The Monopoly Case} \label{sec:mon}

		As a motivation for our general analysis we illustrate in this section our main ideas in the framework of monopolist oil producer. We fix a time 
horizon $T \in (0,\infty]$ and denote the monopolist's resource at time $0$ by $x_0$. The monopolist extracts the resource according to a measurable rate function (control) $q:[0,T] \to [0,\infty)$. A control is called {\em admissible} if the corresponding state process 
\[
	X_t^q :=x_0-\int_0^t q_u\,du
\]
is non-negative on $[0,T]$. Following \cite{chan2014bertrand,ChanSircar15}, we assume that the price function $p:[0,T] \to [0,1]$ when the production rate $q$ is employed is given by $p = 1-q$. For a given constant discount rate $r>0$, the monopolist's value function is then defined by
\begin{equation} \label{value-function-monopoly}
	u(x_0) :=\sup_{q \geq 0 } J(q),
\end{equation}
where 
\[
	J(q) := \int_0^{\tau_q\wedge T}e^{-rt}q_t(1-q_t)\,dt
\]
is the monopolist's  discounted revenue, and
\[
	\tau_q=\tau_q(x_0):=\inf\left\{t\in[0,T]: X^q_t = 0\right\}
\]
denotes the time of depletion of the resource under the control $q$. We put 
\[
	\tau_q = + \infty \quad \mbox{if} \quad \tau_q > T. 
\]

Due to the dependence of $\tau_q$ on the control, the Hamiltonian depends 
on the initial state. This calls for a non-standard approach when applying the maximum principle. We can rewrite the value function as
\begin{equation} \label{value-function-char-func}
	u(x_0)=\sup_{q \geq 0}\int_0^{T}e^{-rt}q_t(1-q_t)\indicator_{\{X^q_t>0\}}\,dt,
\end{equation}
and so the Hamiltonian is given by
\begin{equation*}
	H(x,q,y)=q(1-q)\indicator_{\{x>0\}}-qy. \label{ham1}
\end{equation*}
The maximum principle formally states that if $q^*$ is an optimal control 
and $X^*$ the induced state process, then there exists an {\sl adjoint process} $Y$, to be interpreted as $\dot u(X^*_t)$, that follows the dynamics 
\begin{equation*}
\begin{split}
	-\dot Y_t &= -r Y_t+ \text{``}\partial_x H(X_t^*,q_t^*,Y_t)\text{''}, \quad 0\leq t<T \\
	Y_T & =\text{``}\dot u(X_T^*)\text{''},
\end{split}
\end{equation*}
and that satisfies the maximum condition
\begin{equation} \label{optimal-q}
q^*_t = \argmax_{\bar q \in \mathbb R_+} H(X_t^*,\bar q,Y_t)=\frac{1}{2}\left(1-Y_t \right)^+.
\end{equation}

However, $\partial_x H(x_0,q,y)=0$ for $x_0\neq 0$, while $\partial_x H(0,q,y)$ does not exist. Moreover, the terminal condition of $Y$ is not well defined; it will turn out that the derivative $\dot u(X_T^*)$ does not exist. The way to overcome this problem is to apply the maximum principle only up to the time $\tau^*:=\tau_{q^*}$ where $q^*$ is the candidate optimal strategy. In order to obtain the candidate optimal strategy we will determine the value $Y_{\tau^*}$ of the adjoint process at the time of depletion and then define the adjoint process {\sl on the entire time interval} as
\begin{equation}\label{adjoint}
	Y_t= Y_{\tau^*}e^{-r(\tau^*-t)} \quad \mbox{for} \quad t \in [0,T].
\end{equation} 

We need to distinguish three different cases, depending on whether full exploitation occurs strictly before time $T$, exactly at the terminal time 
or whether full exploitation is not optimal. In what follows, we derive a 
characterization of the candidates for the optimal exploitation time and terminal condition of the adjoint process in terms of the initial resource level.

\begin{itemize}
	\item {$\tau^*< T\leq\infty$.} In this case the maximum condition \eqref{optimal-q} suggests that $Y_{\tau^*}=1$ to ensure that $q^*\equiv0$ in 
$[\tau^*,T]$ and hence that $Y_t=e^{-r(\tau^*-t)}\indicator_{\{t\leq\tau^*\}}$ and
\[
	q^*_t(x_0) = \frac{1}{2} \left( 1- e^{-r(\tau^*-t)}\indicator_{\{t\leq\tau^*\}} \right). 
\]
Then $\tau^*$ is determined by the identity 
\[
	x_0=\int_0^{\tau^*}q_t^*\,dt=\frac{1}{2}\int_0^{\tau^*} (1- e^{-r(\tau^*-t)})\,dt=\frac{\tau^*}{2}-\frac{1}{2r}(1-e^{-r\tau^*}).
\]
This gives
\begin{equation} \label{optimal-tau-monopoly}
	\tau^*(x_0)=2x_0+\frac{1+\W(-e^{-1-2rx_0})}{r},
\end{equation}
where $\W$ denotes the principal branch of the Lambert-W function, defined as the inverse function of $xe^x$, restricted to the range $[-1,\infty)$ and the domain $[-e^{-1},\infty)$. This recovers the formula found by dynamic programming methods for the $T=\infty$ case in \cite[Proposition 
3]{chan2014bertrand}.

	\item{{\bf $\tau^*= T<\infty$.}} In this case we expect the terminal value $Y_T(x_0)$ to be obtained by the identity
\[
	x_0=\int_0^T q_t^*\,dt=\frac{1}{2}\int_0^T (1-Y_T(x_0) e^{-r(T-t)})\,dt=\frac{T}{2}-\frac{Y_T(x_0)}{2r}(1-e^{-rT}),
\]
which yields 
\begin{equation} \label{optimal-hitting-monopoly}
	Y_T(x_0)=\frac{r(T-2x_0)}{1-e^{-rT}}, \quad\mbox{ and }\quad q_t^*(x_0) = \frac12\left(1 - \frac{r(T-2x_0)}{1-e^{-rT}}e^{-r(T-t)}\right).
\end{equation}
Note that $q^*_T>0$ in general, and that $Y\geq0$ only if $x_0\leq T/2$. This suggests that full exploitation is not optimal if $x_0 > T/2$.   

\item{{\bf $\tau^*=+\infty$.}} In this case, we expect the problem to be locally independent of $x_0$, and hence that $Y\equiv0$. This yields $q^*\equiv\frac{1}{2}$ and therefore necessarily $x_0\geq T/2$ and $T<\infty$. 
\end{itemize}

	Using the above obtained terminal values $Y_T(x_0)$ we now {define} the adjoint process $Y_t$ by \eqref{adjoint} for every $t \in [0,T]$. By construction $Y_t \geq 0$ for all $t \in [0,T]$ and $Y_t \geq 1$ for $t \geq \tau^*$. Moreover, $Y_t \equiv 0$ if $\tau^* > T$ or $x_0=T/2$. In particular, for any admissible control $q$, the Hamiltonian satisfies 
\begin{equation}\label{rem-H}
	H(X^q_t,q_t,Y_t) \leq 0 \quad \mbox{on} \quad [\tau^*,T].
\end{equation}

We can now prove that the above constructed candidate optimal strategy is 
indeed optimal.

\begin{theorem}[Verification Theorem] \label{verification-monopoly}
The optimal control at time $t\in[0,T]$ to the control problem \eqref{value-function-char-func} with initial state $x_0>0$ is given by
\[
q^*_t(x_0)=\frac{1}{2}\begin{cases}(1-e^{-r(\tau^*(x_0)-t)})\indicator_{\{t\leq \tau^*(x_0)\}}, &0\leq x_0\leq \frac{T}{2}-\frac{1}{2r}(1-e^{-rT});\\
1-Y_T(x_0)e^{-r(T-t)}, & \frac{T}{2}-\frac{1}{2r}(1-e^{-rT})\leq x_0\leq T/2;\\
1, & x_0\geq T/2;
\end{cases}
\]
where $\tau^*(x_0)$ and $Y_T(x_0)$ are given by \eqref{optimal-tau-monopoly} and \eqref{optimal-hitting-monopoly}, respectively.
\end{theorem}

\begin{proof}
	For notational convenience we drop the dependence of $\tau^*,Y_T$ and $q^*$ on $x_0$. A direct computation verifies that $q^*$ is admissible. Let 
$q$ be another admissible control and $\tau_q$ be the corresponding exhaustion time. We put 
\[	
	\bar\tau_q :=\tau_{q}\wedge T \quad \mbox{and} \quad \bar\tau^* :=\tau^{*}\wedge T. 
\]	
By definition of the adjoint process $e^{-rt}Y_t \equiv Y_0$ on $[0,T]$ where $Y_0 \geq 0$ and $Y_0 = 0$ if $\tau^* > T$. Moreover, for any admissible control $q$ 
\begin{equation}\label{q*1}
	q^*_t =\argmax_{\bar q \in \R_+}H(X^q_t,\bar q,Y_t) \quad \mbox{on} \quad [0,\tau^* \wedge \tau^q].
\end{equation}	
We now distinguish two cases.

\begin{itemize}
\item {The case $\tau^*\geq\tau_q$.} In this case it follows from \eqref{q*1} that
\begin{align*}
	J(q^*)-J(q) &=\int_0^{\bar\tau^*}e^{-rt}q^*_t(1-q^*_t)\,dt-\int_0^{\bar\tau_q} e^{-rt}q_t(1-q_t)\,dt\\
	&=\int_0^{\bar\tau_q}e^{-rt}(H(X^*_t,q^*_t,Y_t)-H(X^q_t, q_t,Y_t)+(q^*_t-q_t)Y_t)\,dt \\ 
	& \quad +\int_{\bar\tau_q}^{\bar\tau^*}e^{-rt}(H(X^*_t, q^*_t,Y_t)+q^*_tY_t)\,dt\\
	&\geq\int_0^{\bar\tau_q}(q^*_t-q_t)Y_te^{-rt}\,dt+\int_{\bar\tau_q}^{\bar\tau^*}q^*_tY_te^{-rt}\,dt\\
 	& = Y_0\left(\int_0^{\bar\tau^*}q^*_t\,dt - \int_0^{\bar\tau_q}q_t\,dt\right)\blk\\	
	&=0
\end{align*}
where the last equality follows from the fact that $Y_0 = 0$ if $\tau^* 
> T$ and that the term in parenthesis vanishes if $\bar \tau^* \leq T$. 

\item {The case $\tau^*<\tau_q$.} In this case, it follows again from \eqref{q*1} that
\begin{align*}
	J(q^*)-J(q) 
	&= \int_0^{\bar\tau^*}e^{-rt}(H(X^*_t,q^*_t,Y_t)-H(X^q_t, q_t,Y_t)+(q^*_t-q_t)Y_t)\,dt \\ 
	& \quad -\int_{\bar\tau_*}^{\bar\tau^q}e^{-rt}(H(X_t, q_t,Y_t)+q_tY_t)\,dt\\
	&\geq\int_0^{\bar\tau^*}(q^*_t-q_t)Y_te^{-rt}\,dt
	-\int_{\bar\tau^*}^{\bar\tau^q}e^{-rt}(H(X^q_t, q_t,Y_t)+q_tY_t)\,dt.
\end{align*}
In view of Equation \eqref{rem-H}, the Hamiltonian is non-positive on $[\tau^*,T]$ and hence 
\[
	J(q^*)-J(q) \geq Y_0 \left( \int_0^{\bar\tau^*}q^*_t \, dt - \int_0^{\bar \tau^q} q_t\,dt \right).
\]
If $\tau^* > T$, then $Y_0=0$. Else, $Y_0 \geq 0$, $\int_0^{\bar \tau^*}q^*_t dt = x_0$ and $\int_0^{\bar \tau^q} q_t\,dt = x_0 - X^q_{\bar \tau^q}$. As a result, $$J(q^*)-J(q) \geq 0.$$
\end{itemize}
\end{proof}

Having determined the optimal control allows us to compute the value function.
	
\begin{corollary}
The value function to the control problem \eqref{value-function-monopoly} 
is given by the smooth function
\[
u(x_0)=\frac{1}{4r}\begin{cases}\left(1+\W(-e^{-1-2rx_0})\right)^2,  &0\leq x_0\leq \frac{T}{2}-\frac{1-e^{-rT}}{2r};\\
1-e^{-rT}-\frac{r^2(T-2x_0)^2}{e^{rT}-1}, & \frac{T}{2}-\frac{1-e^{-rT}}{2r}\leq x_0\leq \frac{T}{2};\\
1-e^{-rT}, & x_0\geq\frac{T}{2}.
\end{cases}
\]
\end{corollary}
This corresponds with the solution to the dynamic programming equation in 
the $T=\infty$ case found in \cite[Section 5]{HHS}.


\section{The Mean Field Game} \label{sec:game}

To motivate the form of demand functions that we are going to use in the continuum MFG, we first introduce a finite market with $N$ oil producers that compete for market share in a one-period game. Associated to each firm $i \in \{ 1, \dots, N \}$ are variables $p_i\in\R$ and $q_i \in \R_+$ representing the price and quantity, respectively. In the Cournot model, players choose quantities as a strategic variable in non-cooperative competition with the other firms, and the market determines the price of each 
good. The market model is specified by linear inverse demand functions, which give prices as a function of quantity produced.
The firms are suppliers, and so quantities are nonnegative.
For $q\in\R^N_+$, the price received by player $i$ is $p_i=P_i(q)$ where
\begin{equation}\label{invdem}
P_i(q) = 1 - \left( q_i + \epsilon \bar q_i \right), \quad \mbox{where}\quad \bar q_i = \frac{1}{N-1} \sum_{j \neq i} q_j,  \,\,i=1,\cdots,N, \quad \mbox{and}\quad 0\leq\epsilon<N-1.
\end{equation}

The inverse demand functions are decreasing in all of the quantities, and 
$\epsilon$ measures the strength of interaction between players.
In the linear model \eqref{invdem}, some of the prices $p_i=P_i(q)$ may 
be negative, meaning player $i$ produces so much that he has to pay to have his goods taken away, but negative prices do not arise in competitive equilibrium. Moreover, the goods are similar but differentiated, meaning each player potentially receives a different price as there will be some residual loyalty to obtaining the product from individual suppliers. Most 
crucially, the interaction is of mean field type: $p_i$ is affected by the mean production of the other players, and players $j$ and $k$ ($j,k\neq 
i$) are exchangeable as far as player $i$ is concerned. 

In the MFG version, there is a continuum of players, say oil producers, who are labeled by their reserves at time. Initial reserves are distributed according to the probability measure $\mu$ on $[0,\infty)$. 
Each producer extracts oil in continuous time at a rate $q_t\geq0$, and the price received by this producer is $P(q_t,Q_t)=1-q_t-\varepsilon Q_t$, where $Q_t$ is the mean production rate of all the players, and $\varepsilon\geq0$ quantifies the degree of interaction between the players. In 
the following, the time horizon $T\leq\infty$.


\subsection{The best response function\label{mfginf}}

In a first step, we consider the representative player's best response {\em given} the aggregate production of all the other players.  
Aggregate production is described by a nonnegative and absolutely continuous function $Q:[0,T]\blk\rightarrow [0,\frac{1}{2+\varepsilon}]$. {The derivative of $Q$ exists a.e.~and is denoted by $\dot Q$.} 

\begin{remark}
It is in fact sufficient here to simply assume  $Q\leq 1/\varepsilon$ which is immediately evident to guarantee nonnegative prices. The a priori bounds $0\leq Q\leq1/(2+\varepsilon)$ are motivated by the following observation: for any candidate optimal control we have 
\[
	0\leq q \leq\argmax_{q\geq0}\{q(1-\varepsilon Q-q)\}=\frac{1}{2}(1-\varepsilon Q)^+,
\]
where the right-hand side is the optimal control given infinite reserves. 
Any aggregate production function $Q$ leading to a solution to the MFG therefore necessarily satisfies 
\[ 0\leq Q\leq \frac{1}{2}(1-\varepsilon Q) \quad \mbox{and so} \quad Q\leq\frac{1}{2+\varepsilon}. \]
\end{remark}

We assume throughout that the aggregate production function satisfies the 
following compatibility condition. We will see that this condition guarantees that each player fully exploits her initial resources if $T$ is large enough. The assumption will be satisfied in equilibrium.  

\begin{assumption}\label{assump1}
There exists $\delta > 0$ such that the aggregate exploitation rate $Q:[0,T]\blk\rightarrow [0,\frac{1}{2+\varepsilon}]$ satisfies the compatibility condition
\begin{equation} \label{simple-cc}
	 1-\varepsilon Q+\frac{\varepsilon}{r} \dot Q \geq \delta >0.
\end{equation}
\end{assumption}

\begin{remark}
Our compatibility condition is equivalent to $\frac{d}{dt}(e^{-rt}Q_t)>\frac1\varepsilon\frac{d}{dt}(e^{-rt})$. Since we expect $Q$ to be decreasing in equilibrium (production slows as resources run out), the condition \eqref{simple-cc} puts a lower bound on how quickly that may occur over time.
\end{remark}

Let us now consider a representative producer with any initial state $x_0\in[0,\infty)$ at time $0$. We call a control $q:[0,T]\to[0,1]$ admissible if the state process
\[
	X_t^q :=x_0-\int_0^t q_u\,du
\]
 is always non-negative.  
 For $T < \infty$ or $T=\infty$, depending on whether the finite or infinite horizon case is considered, 
the value function for the representative producer with respect to a given function $Q$ is defined by
\begin{equation} \label{fixed-Q-control-problem}
	u^Q(x_0)=\sup_{q \geq 0}J^Q(q), \quad \mbox{where }\quad  J^Q(q):=	
	\int_0^{\tau_q^Q\wedge T}e^{-rt}q_t(1-\varepsilon Q_t-q_t)\,dt,
\end{equation}
and the exhaustion time $\tau^Q_q=\tau^Q_q(x_0)$ is 
\[ \tau_q^Q := \inf\left\{t>0\mid X_t^q  = 0\right\}. \]
By analogy to the single player case, the corresponding Hamiltonian is given by 
\begin{equation*}
	H(x,q,y)=q(1- \varepsilon Q - q)\indicator_{\{x>0\}}-qy. \label{hammfg}
\end{equation*}


\subsubsection{Heuristics by the maximum principle}

As in the single player case Pontryagin's maximum principle formally states that if $q^Q$ is an optimal control, with associated {\sl optimal} state process $X^{q^Q}$ and vanishing time  $\tau^Q$, then there exists an adjoint process $Y$ given by 
\begin{equation}
\label{hotelling}
	Y_t = Y_{\tau^Q} e^{-r (\tau^Q-t)}, \quad t \geq 0 
\end{equation}
and $q^Q$ satisfies the maximum condition
\begin{equation} \label{optimal-control-MFG-infinite}
	q^Q_t = \argmax_{\bar q \in \R_+} H(X^{q^Q}_t,\bar q,Y_t)=\frac{1}{2} \left( 1-\varepsilon Q_t-Y_t \right)^+.
\end{equation}
The challenge is again that the terminal condition of the adjoint process 
is unknown. As in the single player case we distinguish again three different cases depending on whether full exploitation is optimal or not. 

\begin{itemize}
	\item {The case $\tau^Q < T\leq\infty$.} In this case we expect $q^Q_{\tau^Q}=0$, and so, from \eqref{optimal-control-MFG-infinite},
the terminal condition of the adjoint process is $Y_{\tau^Q}=1-\varepsilon Q_{\tau^Q}$.
Hence, from \eqref{hotelling}, we have
\[
	Y_t=(1-\varepsilon Q_{\tau^Q})e^{-r(\tau^Q-t)},\quad t<\tau^Q,
\] 
and the optimal strategy is
\begin{equation}
	q^Q_t(x_0) = \frac{1}{2}\left( 1-\varepsilon Q_t- (1-\varepsilon Q_{\tau^Q(x_0)})e^{-r(\tau^Q(x_0)-t)} \right),\quad t<\tau^Q.	\label{qQ}
\end{equation}

\item {The case $\tau^Q = T<\infty$.} 
In this case the terminal value $Y_T(x_0)$ of the adjoint process is obtained by the identity
\[
	x_0=\int_0^T q_t^Q(x_0)\,dt=\frac{1}{2}\int_0^T \left(1-\varepsilon Q_t-Y_T(x_0) e^{-r(T-t)}\right)
	\,dt,
\]
which yields
\begin{equation}
	Y_T(x_0)=\frac{1}{\beta_T}\left(\frac12\int_0^T(1-\varepsilon Q_t)\,dt-x_0\right), \quad \mbox{where} \quad \beta_T := \frac{1-e^{-rT}}{2r}. \label{betadef}
\end{equation}
In particular, the terminal condition depends on the initial resource level and $Y\geq0$ if and only if 
\begin{equation}
x_0\leq \eta^Q(T), \quad \mbox{where} \quad \eta^Q(T):=\frac{1}{2}\int_0^T(1-\varepsilon Q_t)\,dt. \label{etadef}
\end{equation}
In this case, the candidate optimal strategy is
\[
	q^Q_t(x_0) = \frac12\left(1-\varepsilon Q_t-Y_T(x_0) e^{-r(T-t)}\right), \quad t<T.
\]

\item {The case $T<\infty$ and $\tau^Q >T$.} 
In this case we expect the optimization problem to be independent of $x_0$ and hence that $Y \equiv 0$. From \eqref{optimal-control-MFG-infinite}, 
this leads to the candidate optimal strategy
\[
	q^Q_t(x_0) = \frac{1}{2}\left( 1 - \varepsilon Q_t \right).
\] 
Thus, this case occurs when $x_0\geq\eta^Q(T)$.
\end{itemize}


\subsubsection{Verification Result}

In order to carry out the verification argument it will be convenient to define the function 
\begin{equation} \label{xiQ}
	\xi^Q(t) :=\frac{1}{2}\int_0^{t} \{1-\varepsilon Q_s- (1-\varepsilon Q_{t})e^{-r(t-s)}\}\,ds,
\end{equation}
which is the amount of resource extracted up to time $t$ by using the strategy \eqref{qQ} that would be optimal if the exhaustion time $\tau^Q$ was equal to $t$. The function $\xi^Q$ allows us to determine the optimal exploitation time as a function of the initial resource level. In fact the 
derivative of $\xi^Q$ is given by
\begin{equation}\label{derivative-xiQ}
	\dot\xi^Q(t)=\frac{1}{2}\left\{1-\varepsilon Q_{t}+\frac{\varepsilon}{r} \dot Q_{t}\right\}(1-e^{-rt}).
\end{equation}
The compatibility condition \eqref{simple-cc} guarantees that $\dot \xi^Q(t) \geq \frac{\delta}{2}(1-e^{-rt}) > 0$. As a result, the inverse $(\xi^Q)^{-1}$ is well-defined on $[0,\xi^Q(T)]$. By Assumption \ref{assump1}, 
$\xi^Q(\infty)=\infty$. Moreover, we see that
\begin{equation}
	\xi^Q(\tau^Q(x_0)) = \int_0^{\tau^{Q}(x_0)} q^Q_s(x_0) \, ds, \label{rela}
\end{equation}
and hence that $\tau^Q(x_0)$ is determined implicitly by the identity
\[
	\xi^Q(\tau^Q(x_0)) = x_0, \quad \mbox{i.e.} \quad \tau^Q(x_0) = (\xi^Q)^{-1}(x_0) \quad \mbox{for} \quad x_0\leq \xi^Q(T).
\]

The preceding heuristics suggest that an optimal exploitation strategy for a given initial resource level $x_0$ and a given aggregate production function $Q$ that satisfies our compatibility condition can be obtained by 
first computing the optimal exploitation time by inverting the function $\xi^Q$ introduced in \eqref{xiQ} on its range $[0,\xi^Q(T)]$, and then applying the maximum principle on $[0,\tau^Q(x_0)]$.  The following result shows that the above heuristics do indeed give us an optimal control. 

\begin{theorem}[Verification Theorem] \label{verification}
Let $Q:[0,T]\to\mathbb [0,\frac{1}{2+\varepsilon}]$ be an absolutely continuous function that satisfies the compatibility condition \eqref{simple-cc}. Define the function $\xi^Q:[0,T]\to[0,\xi^Q(T)]$ by \eqref{xiQ}, and 
let $\tau^Q=(\xi^Q)^{-1}$.

\begin{itemize}
	\item $T =\infty$. For initial state $x_0\in[0,\infty)$, $\tau^Q(x_0)<\infty$, and the optimal control is given by
\begin{equation} \label{optimal-control-infinite-time}
	q^Q_t(x_0)= \frac{1}{2}\{1-\varepsilon Q_t-(1-\varepsilon Q_{\tau^Q(x_0)})e^{-r(\tau^Q(x_0)-t)}\}\indicator_{\{t\leq \tau^Q(x_0)\}}.
\end{equation}
	\item $T<\infty$. The optimal control for the initial state $x_0\in[0,\infty)$ is given by
\begin{equation} \label{optimal-control-finite-time}
	q^Q_t(x_0)=
	\begin{cases}
	\frac{1}{2}\{1-\varepsilon Q_t-(1-\varepsilon Q_{\tau^Q(x_0)})e^{-r(\tau^Q(x_0)-t)}\}\indicator_{\{t\leq \tau^Q(x_0)\}}, & x_0\in[0,\xi^Q(T)]\\
	\frac12\left(1-\varepsilon Q_t-Y_T(x_0) e^{-r(T-t)}\right)
	& x_0\in[\xi^Q(T),\eta^Q(T)], \\
	\frac{1}{2}(1-\varepsilon Q_t), & \mbox{else},
\end{cases}
\end{equation}
where $Y_T(x_0)$ and $\eta^Q(T)$ are defined in \eqref{betadef} and \eqref{etadef}, respectively.
\end{itemize}
\end{theorem}

\begin{proof}  For notational convenience we drop the dependence of $q^Q$ 
and $\tau^Q$ on the initial state $x_0$. The proof is similar to the monopoly case; we give it for completeness. First, we recall from the argument following \eqref{derivative-xiQ}, that $\tau^Q$ is well-defined. Moreover, it is straightforward to check that i) the candidate $q^Q$ in \eqref{optimal-control-infinite-time} and \eqref{optimal-control-finite-time} are of the form  
\begin{equation} \label{q-mc} 
	q_t^Q=\frac{1}{2}(1-\varepsilon Q_t-Y_t)^+  
\end{equation} 
where the adjoint process $Y$ is defined in \eqref{hotelling} , and ii) are admissible. In particular $q^Q$, and the corresponding state process $X^{q^Q}$ are non-negative. Denoting $\bar\tau_{q^Q} := \tau_{q^Q}\wedge 
T$, we see from the formulas \eqref{optimal-control-infinite-time} and \eqref{optimal-control-finite-time} that $\bar\tau_{q^Q} = \tau^Q\wedge T$. In order to verify the optimality of the strategy $q^Q$  for the first 
case in \eqref{optimal-control-finite-time}, we fix an admissible control 
$q$, whose exhaustion time is $\tau_q$, with the convention $\tau_q=\infty$ if the strategy does not exhaust all of the resource. We denote by $\bar\tau_q:=\tau_q\wedge T$. Admissibility requires that $ \int_0^{\bar\tau_q}q_t\,dt\leq x_0$. Moreover, since $q^Q$ satisfies the maximum condition \eqref{optimal-control-MFG-infinite}, 
\begin{equation} \label{Hamiltonian1}
	H(x,q^Q_t,Y_t) - H(x,q_t,Y_t) \geq 0 \quad \mbox{for all} \quad x>0.  
\end{equation}

\begin{itemize} 
\item The case $\tau_q\leq\tau_{q^Q}$. In this case it follows from \eqref{Hamiltonian1} and $Y_te^{-rt}=Y_0$ that
\begin{align*}
	J^Q(q^Q)-J^Q(q) &=\int_0^{\bar\tau_q}e^{-rt}(H(X^{q^Q}_t,q^Q_t,Y_t)-H(X^q_t,q_t,Y_t)+(q^Q_t-q_t)Y_t)\,dt\\
	&~~ +\int_{\bar\tau_q}^{\bar\tau_{q^Q}}e^{-rt}(H(X^{q^Q}_t,q^Q_t,Y_t)+q^Q_tY_t)\,dt\\
	&\geq\int_0^{\bar\tau_q}(q^Q_t-q_t)Y_te^{-rt}\,dt+\int_{\bar\tau_q}^{\bar\tau_{q^Q}}q^Q_tY_te^{-rt}\,dt\\
 &  = Y_0\left(\int_0^{\bar\tau_{q^Q}}q^Q_t\,dt - \int_0^{\bar\tau_q}q_t\,dt\right)\blk. 
\end{align*}
As in the monopoly case, $Y_0 = 0$ if $\tau^Q > T$. If $\tau_q \leq \tau^Q \leq T$, then the term in parenthesis vanishes. In both cases 
\[
	J^Q(q^Q)-J^Q(q) \geq 0.
\]

\item The case $\tau_q > \tau^Q$. In this case it follows again from \eqref{Hamiltonian1} that
\begin{align*}
	J^Q(q^Q)-J^Q(q) & = \int_0^{\bar\tau_{q^Q}}e^{-rt}(H(X^{q^Q}_t,q^Q_t,Y_t)-H(X^q_t,q_t,Y_t)+(q^Q_t-q_t)Y_t)\,dt\\
	&~~ - \int_{\bar\tau_{q^Q}}^{\bar\tau_{q}}e^{-rt}(H(X^{q}_t,q_t,Y_t)+q_tY_t)\,dt\\	
	&\geq \int_0^{\bar \tau_{q^Q}}  e^{-rt}(q_t^Q-q_t)Y_t(x_0)\,dt -\int_{\bar\tau_{q^Q}}^{\bar\tau_{q}}e^{-rt}(H(X^{q}_t,q_t,Y_t)+q_tY_t)\,dt.\\	
\end{align*}
Using that $Y\geq 1$ on $[\tau^{q^Q},T]$, and hence that $H\leq0$ for all 
$(q,t)\in\mathbb [0,\infty)\times[\tau^{q^Q},T]$, we arrive at
\[
	J^Q(q^Q)-J^Q(q)  \geq Y_0 \left( \int_0^{\bar \tau_{q^Q}} q_t^Q \,dt - \int_0^{\bar \tau_{q}} q_t \,dt \right).  	
\] 

If $\tau^{q^Q} > T$, then $Y_0=0$. Else, $Y_0 \geq 0$, $ \int_0^{\bar \tau_{q^Q}} q_t^Q \,dt = x_0$ while $\int_0^{\bar \tau_{q}} q_t \,dt \leq x_0$ because $\tau_q > \tau^Q$. As a result, 
\[
	J^Q(q^Q)-J^Q(q) \geq 0.
\]
\end{itemize}
\qedhere
\end{proof}
It follows from the verification theorem that all the optimal production rates $q_t^Q$ satisfy the same ODE, albeit on possibly different time horizons and with possibly different terminal conditions. Moreover, from \eqref{derivative-xiQ}, the mapping $x_0 \mapsto \tau^Q(x_0)$ is strictly increasing and so $$\{x_0 \geq 0 : \tau^Q(x_0) < t\} = \{x_0 \geq 0 : x_0 
< \xi^Q(t)\}.$$ In particular, if we denote by $S_\mu(x_0)=\int_{x_0}^\infty\mu(dx)$ the survival function associated with the initial distribution $\mu$,
then the proportion of players with remaining resources at time $t \in [0,T]$ is given by $S_\mu(\xi^Q(t))$. With this, we have the following corollary.

\begin{corollary}\label{cor-Q}
For all $x_0 > 0$ and all functions $Q$ that satisfy the compatibility condition \eqref{simple-cc}, the optimal rate $q^Q_t(x_0)$  and the aggregate production rate \begin{equation}Q_t^Q: = \int_{\xi^Q(t)}^\infty q^Q_t(x_0)\,d\mu(x_0), \quad 0\leq t \leq T \label{fxpt}
\end{equation}
satisfy the ODEs
\begin{align} 
	q^Q_t(x_0)-\frac{1}{r}\dot q^Q_t(x_0)&=\frac{1}{2}\left\{1-\varepsilon 
Q_t+\frac{\varepsilon}{r}\dot Q_t\right\} \label{ODE-q}\\
Q_t^Q-\frac{1}{r}\dot Q_t^Q&=\frac{1}{2}\left\{1-\varepsilon Q_t+\frac{\varepsilon}{r}\dot Q_t\right\}S_\mu(\xi^Q(t)), \label{terminal-q}
\end{align}	
on the time interval $[0,\tau^Q(x_0)\wedge T]$.
In the case $T = +\infty$, the terminal condition for \eqref{ODE-q} is $q^Q_{\tau^Q(x_0)} =0$, while for \eqref{terminal-q}, $\lim_{t \to \infty} Q_t = 0$.
If $T < \infty$, then the optimal production rate has terminal condition 
\begin{equation} \label{terminal-q1}
	q^Q_{\tau^Q(x_0) \wedge T}(x_0) =
	\begin{cases}
		0 & \mbox{ if } x_0\in[0,\xi^Q(T)] \\
		\frac12\left(1-\varepsilon Q_T-Y_T(x_0)\right)
		& \mbox{ if } x_0\in[\xi^Q(T), \eta^Q(T))] \\
		\frac{1}{2}(1 - \varepsilon Q_T) & \mbox{ else,}
	\end{cases}
\end{equation}
where $Y_T(x_0)$ was given in \eqref{betadef}. 
In this case, the terminal condition for the aggregate production rate is 

\begin{equation}
	Q^Q_T =\frac12(1 - \varepsilon Q_T)S_\mu(\xi^Q(T)) - \frac{\eta^Q(T)}{\beta_T}\left(S_\mu(\xi^Q(T)) - S_\mu(\eta^Q(T))\right)
	+\frac{1}{\beta_T}\int_{\xi^Q(T)}^{\eta^Q(T)}x_0\mu(dx_0). 	\label{QTeqn}
\end{equation}
\end{corollary}
\begin{proof}
The ODE \eqref{ODE-q} follows from differentiating the expression \eqref{q-mc} for $q^Q$ (the positive part being unnecessary on $t\in[0,\tau^Q(x_0)\wedge T]$), using the ODE for $Y$, and substituting back for $Y$ in terms of $Q$ and $q^Q$. The ODE \eqref{terminal-q} for $Q^Q$ follows from integrating \eqref{ODE-q} with respect to $\mu(x_0)$ over $(\xi^Q(t),\infty)$. The terminal conditions for $q^Q$ follow from \eqref{optimal-control-infinite-time} and \eqref{optimal-control-finite-time}. The terminal condition for $Q^Q$ comes from integrating \eqref{terminal-q1} with respect to $\mu$ over $(0,\infty)$.
\end{proof}
\begin{remark}
Under the compatibility condition, we have $\xi^Q(T) \uparrow \infty$ as $T \uparrow \infty$. As a result, $\lim_{T \to \infty} Q^Q_T = 0$. This 
shows that the infinite horizon case can indeed be viewed as a limiting case when $T \uparrow \infty$. 
\end{remark}


\subsection{Solution to the MFG } \label{sec-solution-to-MFG}

In this section we prove the existence of a solution to the mean field game. In terms of the optimal controls $q^Q$ we are looking for a fixed point of the mapping \eqref{fxpt}. We consider the infinite and the finite horizon case separately. The main difference is that full exploitation occurs in equilibrium if $T=+\infty$, but may not occur if $T < \infty$. It turns out that the infinite horizon case can be solved in closed form. In view of Corollary \ref{cor-Q} , we have the following characterization.
\begin{corollary}\label{cor-compatibility}
Any fixed point $Q^*=Q^{Q^*}$ satisfies the dynamics   
\begin{equation*}
	Q^*-\frac{1}{r}\dot Q^*=\frac{S_\mu\circ \xi^{Q^*}}{2+\varepsilon S_\mu\circ \xi^{Q^*}}, \label{fxdptODE}
\end{equation*}
and satisfies the compatibility condition \eqref{simple-cc}. 
\end{corollary}
\begin{proof}
Setting $Q^Q$ and $Q$ in \eqref{terminal-q} to $Q^*$ and re-arranging leads to 
\eqref{fxdptODE}, 
which is equivalent to 
\begin{equation} \label{key-observation}
1-\varepsilon Q^*+\frac{\varepsilon}{r}\dot Q^*=\frac{2}{2+\varepsilon S_\mu\circ \xi^{Q^*}}\,,
\end{equation}
from which compatibility of $Q^*$ easily follows.
\end{proof}

For any fixed point $Q^*$ the equations \eqref{derivative-xiQ}  and \eqref{key-observation} yield the following ODE for $\xi^{Q^*}$: 
\begin{equation}\label{ODE-xiQ}
	\displaystyle\dot\xi^{Q^*}(t) = \frac{1-e^{-rt}}{2+\varepsilon S_\mu\circ \xi^{Q^*}}, \quad \xi^{Q^*}_t(0) = 0. 
\end{equation}
The key observation is that \eqref{ODE-xiQ} does not depend on $Q^*$. We prove below that the (unique) solution to the MFG can be defined in terms 
of the unique solution $\xi^*$ to the above ODE. 

\begin{proposition}
If the initial distribution has a bounded density or is concentrated on finitely many points, then the preceding ODE has a unique global solution. 
The solution can be expressed in terms of the functions $\phi,\psi:[0,\infty)\to[0,\infty)$ defined by
\[
	\phi(t):=\frac{rt-1+e^{-rt}}{r} \quad \mbox{and} \quad \psi(x):=2x+\varepsilon\int_0^xS_\mu(\xi)\,d\xi,
\]
respectively as
\begin{equation}
	\xi^*:=\psi^{-1}\circ\phi
\end{equation}
\end{proposition}
\begin{proof}
If $\mu$ has a bounded density, then the mapping $x \mapsto \frac{1-e^{-rt}}{2+\varepsilon S_\mu(x)}$ is Lipschitz continuous. If $\mu$ is concentrated on finitely many points, then this mapping is locally independent of $x$ and hence can be solved iteratively. The functions $\phi$ and $\psi$ are well-defined, strictly increasing, surjective, and satisfy
\[
	\dot\phi(t)=1-e^{-rt} \quad \mbox{and} \quad \psi'(x)=2+\varepsilon S_\mu(x).
\]
Now, a direct computation verifies that $\xi^*$ satisfies the desired ODE.
\end{proof}
Note that the inverse of $\phi$ is, in terms of the principle branch of the Lambert-$\W$ function,
\[
	\phi^{-1}(x)=x+\frac{1+\W(-e^{-1-rx})}{r}.
\]


\subsubsection{Infinite horizon}

We are now ready to solve the MFG in the infinite horizon case. In terms of the functions $\phi$ and $\psi$, and their inverses we obtain 
\begin{itemize}
	\item $\xi^*:=\psi^{-1}\circ\phi$
	\item $\tau^*:={\xi^*}^{-1}=\phi^{-1}\circ\psi$
\end{itemize}
In view of \eqref{key-observation} we can express the unique solution to the ODE \eqref{ODE-q} with terminal condition \eqref{terminal-q} fully in 
terms of these functions as  
\[
	\displaystyle q_t^*(x_0):=\int_{t\wedge\tau^*(x_0)}^{\tau^*(x_0)}\frac{re^{-r(s-t)}}{2+\varepsilon S_\mu(\xi^*(s))}\,ds.
\]

The following theorem verifies that this production rate does indeed solve the MFG in the infinite horizon case. 

\begin{theorem}
The aggregate production rate 
\[
	\displaystyle Q_t^*:=\int_0^\infty q_t^*(x_0)\,d\mu(x_0)
\]
is absolutely continuous, takes values in $[0,\frac{1}{2+\varepsilon}]$ and satisfies the compatibility condition \eqref{simple-cc} as well as the 
fixed point property $Q^{Q^*}=Q^*$. 
Hence, $Q^*$ is the unique solution to the MFG (that is absolutely continuous and satisfies the compatibility condition).
\end{theorem}
\begin{proof} 
	To verify $Q^*\in[0,\frac{1}{2+\varepsilon}]$ we use the following representation by Fubini: 
\begin{equation} \label{fubini-representation} 
	\begin{aligned}
	Q_t^*&=\int_t^\infty \frac{r S_\mu(\xi^*(s))e^{-r(s-t)}}{2+\varepsilon 
S_\mu(\xi^*(s))}\,ds
	\leq \frac{S_\mu(\xi^*(t))}{2+\varepsilon S_\mu(\xi^*(t))}
	\leq \frac{1}{2+\varepsilon}.
	\end{aligned}
\end{equation}
Next, by construction $\xi^*$ satisfies the ODE 
\begin{equation} \label{eq-diff-xi}
	\dot \xi^*(t)=\frac{1-e^{-rt}}{2+\varepsilon S_\mu(\xi^*(t))}, \quad \xi^*(0) = 0
\end{equation}
and $q^*$ is absolutely continuous in the time variable with
\begin{equation} \label{eq-diff-q}
	q^*_t(x_0)-\frac{1}{r}\dot q_t^*(x_0)=\frac{1}{2+\varepsilon S_\mu(\xi^*(t))}\indicator_{\{t<\tau^*(x_0)\}}.
\end{equation}
The latter yields that $Q^*$ is absolutely continuous with
\begin{equation} \label{strict-cc}
	Q_t^*-\frac{1}{r}\dot Q_t^*=\frac{S_\mu(\xi^*(t))}{2+\varepsilon S_\mu(\xi^*(t))},
\end{equation}
which is equivalent to 
\[
	\frac{1}{2}\left\{1-\varepsilon Q_t^*+\frac{\varepsilon}{r}\dot Q_t^*\right\}=\frac{1}{2+\varepsilon S_\mu(\xi^*(t))}.
\]
As a result, $Q^*$ satisfies the compatibility condition. Moreover, in view of \eqref{derivative-xiQ} and \eqref{eq-diff-xi}, 
\[
	\xi^* = \xi^{Q^*} \quad \mbox{and hence} \quad \tau^* = \tau^{Q^*}. 
\]
>From this we conclude that $q^{Q^*}=q^*$ and hence the fixed point property $Q^{Q^*}=Q^*$. Applying the verification result \cref{verification} finally verifies that $Q^*$ is the unique solution the MFG.
\end{proof}

While the individual production rate $q^*$ may not be monotone in general 
as illustrated below, it is a direct consequence of \eqref{fubini-representation}  and \eqref{strict-cc} that the aggregate production rate is nonincreasing. 
\begin{corollary}
The aggregate production rate $Q^*$ is nonincreasing.
\end{corollary}


\subsubsection{Finite horizon}

In the finite horizon case an additional challenge emerges. While we can still solve the ODE \eqref{ODE-q} using \eqref{key-observation}, due to the dependence of the second and third terminal conditions in \eqref{terminal-q1} on $Q_T$ we obtain the equilibrium production rate only up to its 
terminal value.\footnote{In the infinite horizon game both the dynamics and the terminal conditions were defined in terms of $\xi^*$ and its inverse $\tau^*$. This is no longer the case in the finite horizon game.}  An additional fixed point argument on the terminal value $Q^*_T$ is required 
to solve the game. In terms of the yet to be determined terminal condition $Q^*_T$ we do know that  
\[
	Q^*_t=Q_T^*e^{-r(T-t)}+\int_t^T \frac{r S_\mu(\xi^*(s))e^{-r(s-t)}}{2+\varepsilon S_\mu(\xi^*(s))}\,ds.
\]
In terms of
\begin{equation} \label{bpsi}
	\psi^Q(x_0) := \frac{1}{\beta_T}\left( x_0 - \xi^Q(T) \right) 
\end{equation}
we obtain for any admissible process $Q$,
\begin{align*}
	Q_T^Q&=\frac{1}{2}\int_{\xi^Q(T)}^{\xi^Q(T)+\beta_T(1-\varepsilon Q_T)}\psi^Q(x_0)\,d\mu(x_0)+\frac{1}{2}\int^\infty_{\xi^Q(T)+\beta_T(1-\varepsilon Q_T)}(1-\varepsilon Q_T^Q)	\,d\mu(x_0)\\
	&=\frac{1}{2}\int_{\xi^Q(T)}^\infty \{\psi^Q(x_0)\wedge(1-\varepsilon Q_T^Q)\}\,d\mu(x_0).
\end{align*}
Putting 
\[
	\psi^*(x_0) := \frac{1}{\beta_T}\left( x_0 - \xi^*(T) \right) 
\]
we obtain for any aggregate equilibrium production rate function $Q^*$ that
\begin{align*}
	Q_T^*=\frac{1}{2}\int_{\xi^*(T)}^\infty \{\psi^*(x_0)\wedge(1-\varepsilon Q_T^*)\}\,d\mu(x_0).
\end{align*}
This means that $Q_T^*$ is a fixed point of the map $\Gamma:[0,\frac{1}{2}]\to[0,\frac{1}{2}]$ definied by
\[
\Gamma(\mathcal Q)=\frac{1}{2}\int_{\xi^*(T)}^\infty\{(\psi^*(x_0)\wedge\{1-\varepsilon \mathcal Q\})\vee0\}\,d\mu(x_0),
\]
where we added the nonnegative cut-off (being redundant for any fixed point) to obtain a map on $[0,\frac{1}{2}]$. Moreover, is $\Gamma$ nonincreasing and continuous. Hence, it admits a unique fixed point $\mathcal Q^*=\Gamma(\mathcal Q^*)$. Since
\[
\Gamma(\mathcal Q^*)\leq \frac{1}{2} S_\mu(\xi^*(T))(1-\varepsilon \mathcal Q^*),
\]
it follows that
\[
\mathcal Q^*\leq \frac{S_\mu(\xi^*(T))}{2+\varepsilon S_\mu(\xi^*(T))}.
\]
This is important to verify that the induced $Q^*$ takes indeed values in 
$[0,\frac{1}{2+\varepsilon}]$. For $t\in[0,T]$,
\begin{align*}
	Q^*_t&=\mathcal Q^*e^{-r(T-t)}+\int_t^T \frac{r S_\mu(\xi^*(s))e^{-r(s-t)}}{2+\varepsilon S_\mu(\xi^*(s))}\,ds\\
	&\leq\frac{S_\mu(\xi^*(t))}{2+\varepsilon S_\mu(\xi^*(t))}e^{-r(T-t)}+\frac{S_\mu(\xi^*(t))}{2+\varepsilon S_\mu(\xi^*(t))}\int_t^T re^{-r(s-t)}\,ds\\
	&=\frac{S_\mu(\xi^*(t))}{2+\varepsilon S_\mu(\xi^*(t))}.
\end{align*}
The fact that $Q^*$ satisfies the compatibility condition has already been established in Corollary \ref{cor-compatibility}. Let us summarize.
\begin{theorem}
In terms of the above definitions, the function $Q^*$ is absolutely continuous, takes values in $[0,\frac{1}{2+\varepsilon}]$, and satisfies the compatibility
condition \eqref{simple-cc} and the fixed point property 
\[
	Q^{Q^*}=Q^*.
\]
Hence, $Q^*$ is the unique solution to the finite-time MFG (that is absolutely continuous and satisfies the compatibility condition).
\end{theorem}


\subsection{Examples}

We close this section with two examples with infinite time horizon that illustrate that equilibrium production rates need not be monotone, nor smooth. 

\subsubsection{Discrete initial distribution}

Let us assume that the population of producers splits into two groups within which producers are identical. Producers in Group 1 have an initial resource $x_1$; producers in Group 2 have an initial resource $x_2$. That is, the initial distribution is given by $$\mu=(1-p_2)\delta_{x_1}+p_2\delta_{x_2}$$ for $0\leq  x_1\leq x_2<\infty$ and $0\leq p_2\leq 1$ where 
$p_2$ denotes the proportion of producers that belong to Group 2. We refer to a representative producer in Group $i$ as Player $i$ where $i=1,2$. The equilibrium production rates can be computed in closed form. The equilibrium production rate of Player 1 is given by
\begin{equation} \label{single-player}
	q_t^*(x_1)=\frac{1-e^{-r(\tau_1-t)}}{2+\varepsilon}\indicator_{\{t\leq\tau_1\}},
\end{equation}
where $\tau_1:=\phi^{-1}((2+\varepsilon)x_1)$, while the equilibrium production rate of Player 2 is given by
\[
	q_t^*(x_2)=\begin{cases}\dfrac{1-e^{-r(\tau_1-t)}}{2+\varepsilon}+\dfrac{e^{-r(\tau_1-t)}-e^{-r(\tau_2-t)}}{2+\varepsilon p_2}, & t\leq \tau_1,\\
	\dfrac{1-e^{-r(\tau_2-t)}}{2+\varepsilon p_2}\indicator_{\{t\leq \tau_2\}}, & t\geq \tau_1,\end{cases}
\]
where $\tau_2:=\phi^{-1}(\varepsilon (1-p_2)x_1+(2+\varepsilon p_2)x_2)$. We notice that the equilibrium rate of Player 2 is continuous, but not 
differentiable at time $\tau_1$. Figure \ref{figure-kink} illustrates the 
equilibrium production rates $q^*_t(x_i)$ and resources $X^*_t(x_i)$ of both players. While both Players initially produce at the same rate, Player 1 produces at a decreasing rate while Player 2 initially produces at an 
increasing rate, and then at a decreasing rate once Player 1 has run out of resources. Player 2 anticipates the fact Player 1 will eventually drop 
out of the market; her production rate reaches it peak at the time Player 
1's resources have been depleted. 

\begin{figure}[h]
	\centering
	\includegraphics[width=.45\textwidth]{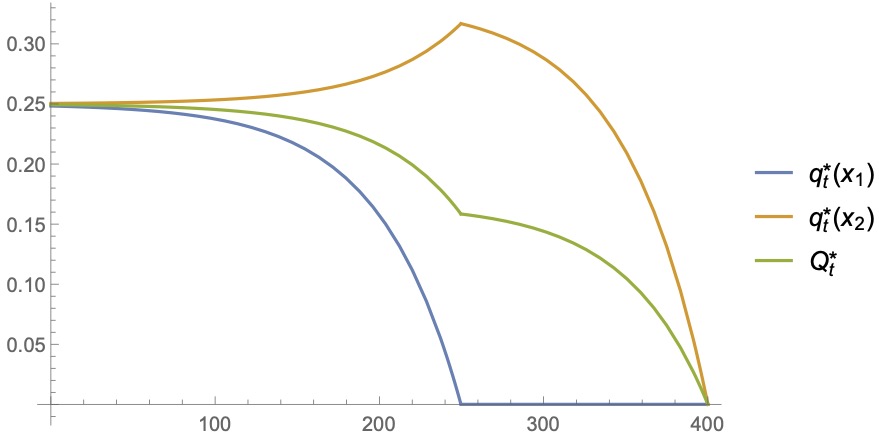}
\hfill	\includegraphics[width=.45\textwidth]{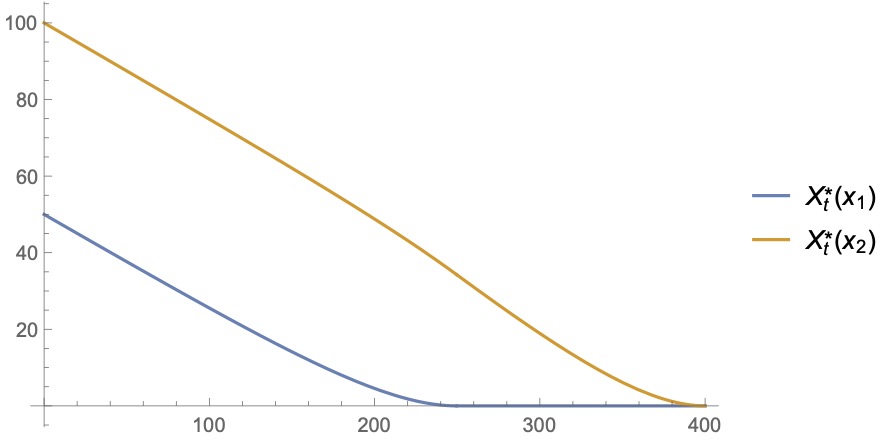}
	\caption{$r=.02$, $\varepsilon=2$, $x_1=50$, $x_2=100$, $p_2=.5$}
	\label{figure-kink}
\end{figure}


\subsubsection{Exponential initial distribution}

Figure \ref{figure-exponential} illustrates the equilibrium production rates $q^*_t(i)$ for players with initial resources $x_i = 1,2,3$ for the 
benchmark case of an exponential initial distribution. For absolutely continuous initial distributions players ``gradually'' drop out of the market so equilibrium production rates do not display kinks. However, we still 
observe non-monotonicity of production rates. Players with larger resources do again anticipate that players with lower reserves will eventually drop out of the game. For them it is optimal to produce at higher rates once the number of and the price pressure from their competitors decreases.

\begin{figure}[h]
	\centering
\includegraphics[width=.6\textwidth]{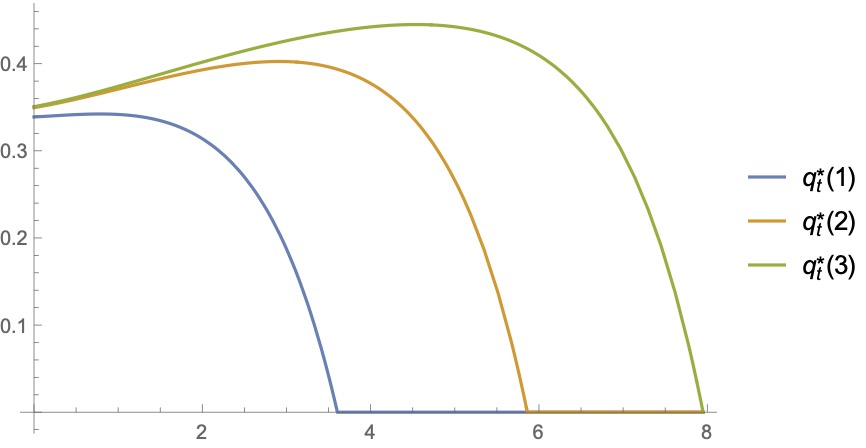}
	\caption{$\mu\sim \text{Exp}(\lambda)$, $r=\varepsilon=\lambda=1$ }
	\label{figure-exponential}
\end{figure}

	\bibliographystyle{plain}
	\bibliography{QQ1}

\end{document}